\documentclass[11pt]{article}
\usepackage{amsmath, amssymb, amsthm}
\usepackage{fullpage}
\usepackage{hyperref}

\hypersetup{
    colorlinks=true,
    linkcolor=blue,
    citecolor=blue,
    urlcolor=blue
}

\usepackage{graphicx}
\usepackage{bbm}

\usepackage{amsmath}
\usepackage{amssymb}
\usepackage{mathtools}
\usepackage{algorithm}
\usepackage{algorithmic}
\usepackage{enumitem}
\usepackage{natbib}

\usepackage{pgfplots}
\usepackage{tikz}
\usetikzlibrary{matrix, positioning, patterns}
\pgfplotsset{compat=1.18}

\newcommand{\E}{\mathop{\mathbb{E}}}

\usepackage{anyfontsize}

\usepackage{ebgaramond}

\newtheorem{theorem}{Theorem}
\newtheorem{lemma}[theorem]{Lemma}
\newtheorem{corollary}[theorem]{Corollary}
\newtheorem{definition}[theorem]{Definition}

\newtheorem{remark}[theorem]{Remark}

\title{Compression Barriers in Autoregressive Transformers}
\author{Themistoklis Haris \\
Boston University \\
\texttt{tharis@bu.edu}
\and
Krzysztof Onak\\
Boston University \\
\texttt{konak@bu.edu}
}

\begin{document}

\maketitle

\begin{abstract}
    A key limitation of autoregressive Transformers
    is the large memory needed at inference-time to cache all previous key-value (KV) embeddings.
    Prior works address this by compressing the KV cache but often assume specific structural properties of the embeddings.
    This raises the following natural question: Can truly sublinear space utilization be achieved without such assumptions?
    In this work, we answer this question in the negative. Any algorithm for attention-based token generation must use $\Theta(nd)$ space, where $n$ is the number of tokens generated so far and $d \geq \Omega(\log n)$ is the dimension of the KV embeddings.
    %
    Our proof involves a reduction from a classic communication complexity problem and uses a randomized construction that leverages
    properties of projections in the spirit of the Johnson-Linderstrauss lemma.
    For the low-dimensional regime $d = o(\log n)$, we show that any algorithm requires $\Omega(de^d)$ space and prove, using tight bounds on covering numbers, that \textsc{SubGen}, proposed by \citet{Zandieh2024}, matches this bound. 
    Further, we investigate how sparsity assumptions enable token generation in truly sublinear space, presenting impossibility results and proposing a new KV cache compression algorithm for sliding window attention when the value cache outside the window is unmasked. 
    Finally, we analyze token generation's time complexity, using an indistinguishability argument to prove that no non-adaptive algorithm can compute attention online in sublinear time for all tokens.
\end{abstract}

\section{Introduction}
Transformers have revolutionized the field of Deep Learning by enabling the creation of models that achieve state-of-the-art performance on a wide range of tasks \citep{dosovitskiy2020image,vaswani2017attention}. 
Central to their success lies the attention mechanism, which allows the model to ``focus'' on different parts of the input sequence when generating the output. 
However, the attention mechanism comes with a significant drawback: the large memory footprint required at token generation for caching numerous embedding vectors for each previous token. 
The cached token embeddings are known as the \textit{key-value (KV) cache}. 
Indeed, due to the KV cache, the memory required by Transformer LLMs for long-context token generation grows linearly with the number of tokens, which can be prohibitive for large prompts.

Numerous past works \citep{Jin2024,liu2024scissorhands,zandieh2024qjl,Zhang2023,liu2023deja,Zandieh2024,chen2021scatterbrain,liu2024minicache} have attempted to address this issue by compressing the KV cache to achieve sublinear space utilization. 
These algorithms typically rely on specific structural assumptions about the key-value token embeddings.
For example, the sliding window attention mechanism assumes that each token attends only to the $W$ tokens before it, masking away both the key and value embeddings outside of a short sliding window. 
This assumption allows for sublinear space utilization by storing only the key and value vectors in the sliding window.

This leads to the following natural question:
\begin{center}
    \textit{Is there a sublinear space algorithm for attention-based autoregressive token generation\\ on 
    inputs that lack significant structural assumptions?}
\end{center}

In this paper, we answer the above question in the negative by showing that any algorithm for attention-based token generation must use linear space in the number of tokens to store the KV cache, even if sparsity is present.
This result suggests that KV cache compression should probably best be viewed from a data-driven perspective, with the most successful algorithms being the ones that quickly and effectively exploit, learn and adapt to the structure of the data.
Beyond this insight, our work also draws on techniques from the extensive literatures on streaming, property testing, sublinear algorithms, communication complexity, and dimensionality reduction, demonstrating their effectiveness and utility 
in complexity theory, learning, and optimization.


\subsection{Related Work}
\paragraph{KV Cache Compression}
A large body of work has recently focused on KV cache compression for Transformer LLMs. 
These algorithms exploit either structural or learned sparsity patterns in the input embeddings.
\citet{Zhang2023}, for instance, proposed an eviction scheme for the KV cache that leverages techniques from dynamic submodular maximization. 
\citet{xiaoefficient} observed empirically that the initial tokens in the input sequence---which they called \textit{attention sinks}---often retain high attention weights throughout the entire generation process. 
This suggests a natural way to compress the KV cache by storing these sinks, in addition to a sliding window of recent token embeddings. 
In a similar fashion, \citet{liu2023scissorhands} observe that tokens with high attention weights tend to remain highly influential to the attention output, a phenomenon they call \textit{persistence of importance}. 
The work of \citet{liu2023deja} has also been influential in this direction, proposing the importance of \textit{contextual sparsity} in Transformer models. 

Beyond the assumptions of sparsity, the work of \citet{Zandieh2024} leverages \textit{clusterability} properties of the key embeddings to achieve provably sublinear space in certain parameter regimes. 
More hardware-oriented methods have also been proposed, such as the works of \citet{Duanmu2024} and \citet{zandieh2024qjl} which leverage \textit{quantization}, and the work of \citet{Jin2024} which examines the I/O complexity of the problem.

\paragraph{Impossibility Results for Transformers}
As research towards efficient Transformers progresses, complexity-theoretic techniques have helped shed light on their limitations.
The computation of the full attention matrix has been shown to require $\Omega(n^2)$ time via fine-grained complexity reductions \citep{Keles2022,alman2024fast}. 
Similar impossibility results have also been shown for the attention gradients \citep{alman2024fine}.
Furthermore, important insights have been obtained by studying the limitations of Transformers as a computational model through the lens of communication complexity and circuit complexity \citep{alman2024fundamental,chen2024circuit}. 
Finally, the representational and generalization power of Transformers has been theoretically studied as well \citep{sanford2024transformers,sanford2024representational,li2023transformers}.

\subsection{Our Contributions and Results}
We provide a theoretical analysis of the space and time complexity of token generation in Transformers.
During inference, an autoregressive Transformer uses the existing sequence of tokens it generated to produce the next token. 
Formally, given a stream of $d$-dimensional token embeddings $\{(q_i,k_i,v_i)\}_{i=1}^n$ the algorithm calculates at each timestep the attention function $\text{Attn}(q_i, K_i, V_i)$, where $K_i$ and $V_i$ contain all previously seen embedding vectors.
The attention function (see Section \ref{sec:prelims})  considers a normalized inner-product similarity metric between $q_i$ and all the vectors in $K_i$, weighting the vectors $v_i$ based on that metric:
\begin{align}
    Z_i:=\text{Attn}(q_i, K_i, V_i) := \sum\limits_{\ell=1}^n w(q_i, k_\ell)\cdot v_{\ell},\,\text{ where }w(q_i,k_\ell)\propto e^{q_i^T k_\ell}
\end{align}
Our main theorem is a space lower bound for algorithms that approximate the attention function during token generation. 
%
\begin{theorem}
\label{thm:lower-bound}
Let $\eta \in (0,1)$ be an arbitrary constant and $d\geq 2$. Let $Z_n := \operatorname{Attn}(q_n,K_n,V_n) \in \mathbb{R}^d$. Any algorithm that can, with probability at least $9/10$, for all timesteps $i \leq n$, produce a $(1\pm \eta)$-approximation\footnote{We say that $x \in \mathbb{R}^d$ is a $(1\pm\eta)$-approximation of $y \in \mathbb{R}^d$ if for all $j\in[d]$ we have $|x_i-y_i| \leq \eta\cdot y_i$.} $\mathcal{O}_i$ of $Z_n$
must use at least $\Omega(d\cdot\min\{n,e^d\})$ bits of memory.
\end{theorem}

We establish this result by constructing a reduction from a well-known \textit{communication complexity} problem, utilizing \textit{Johnson-Linderstrauss (JL)} random projections and the probabilistic method to create a communication protocol. Our proof relies on the malleability of the softmax distribution, as we use it to approximate a ``Dirac''-like spike on a discretized support set. A similar proof technique was used by \cite{sanford2024representational}, where they used the restricted isometry property from the theory of compressed sensing to analyze the representational capabilities of Transformer models. 

Theorem \ref{thm:lower-bound} is tight for the low-dimensional regime as well: we show that the \textsc{SubGen} algorithm of \citet{Zandieh2024} achieves optimal space utilization for $d = o(\log n)$  (Theorem \ref{thm:subgen-optimal} and Corollary \ref{cor:lower-bound-low-dim}). Our proof relies on a well-known bound for the covering number of the $\ell_2$-ball to give a tight characterization of the clusterability of low-dimensional embedding spaces. Our results imply that space-efficient token generation is only possible if the embedding dimension $d$ scales sub-logarithmically with the number of tokens $n$, a regime in which Transformers unfortunately start to lose their expressivity \citep{sanford2024representational}.

We also investigate the impact of \textit{sparsity} on the space complexity of KV cache compression. 
First, we formally demonstrate that \textit{unstructured sparsity} cannot achieve sublinear space utilization (Corollary \ref{cor:sparsity-not-enough}), a surprising observation given the widespread success of sparsity-based algorithms in the literature. 
Next, we examine the \textit{sliding window} attention mechanism, a very popular fixed sparsity pattern, in a more general setting, where value vectors outside the sliding window are not assumed to be masked. This scenario extends beyond the typical assumptions in the literature, and no sublinear space algorithm has been proposed for it so far.
In this setting, we present a novel algorithm (Algorithm \ref{alg:sliding-window}) for KV cache compression, achieving sublinear space utilization under mild assumptions. We establish that our algorithm is nearly space-optimal by proving a matching lower bound (Theorem \ref{thm:lower-bound-sliding-window}).

Finally, we study the time complexity of non-adaptive streaming algorithms, which are algorithms that pre-decide their data access pattern before seeing the stream. Such algorithms are common in the literature, as they are more straightforward to implement (e.g. fixed sparsity patterns).
%
%

\begin{theorem}
    Let $\mathcal{A}$ be a streaming algorithm for attention-based token generation that decides before seeing any data which entries of $QK^T$ to access.
    Suppose that $\mathcal{A}$, with probability at least $9/10$, outputs for each $i \in [n]$ an $(1\pm\eta)$-approximation $\mathcal{O}_i \in \mathbb{R}^d$ of $Z_i := \operatorname{Attn}(q_i, K_i, V_i)$. Then $\mathcal{A}$ must take $\Omega(nd)$ time to process the last token $(q_n,k_n,v_n)$ in the worst case.
\end{theorem}

We show this result by using an \textit{indistinguishability} argument, with which we show that sampling from a softmax distribution requires $\Omega(n)$ queries to the distribution in the worst case. 
Such techniques are very widespread in the literature of sublinear algorithms and testing.

\section{Preliminaries}
\label{sec:prelims}
\subsection{The Attention Mechanism and KV Caching}
We can formalize the KV Cache Compression problem as follows: The input is a stream of triples $(q_i,k_i,v_i) \in (\mathbb{R}^d)^3$, where $d$ is the embedding dimension. After each stream entry, define the \textbf{Attention function} as:
\begin{align}
    \text{Attn}(q_i,K_i,V_i) := \sigma_i(K_i\cdot q_i)^T \cdot V_i \in \mathbb{R}^d
\end{align}
where:
$$
K_i = \left[\begin{matrix}
    k_1^T\\
    k_2^T\\
    \cdots\\
    k_i^T
\end{matrix}\right]\,\text{ and }\, V_i = \left[\begin{matrix}
    v_1^T\\
    v_2^T\\
    \cdots\\
    v_i^T
\end{matrix}\right]
$$
are two $i\times d$ matrices, and $\sigma_i :\mathbb{R}^i\to \mathbb{R}^i$ is the softmax function with support $[i]$. The $K_i$ and $V_i$ matrices are called the \textbf{key-value (KV) cache}.
The attention function can be viewed as a collection of expected values under suitable softmax distributions \citep{kratsios2021universal,singh2023attention,haris2024knnattentiondemystifiedtheoretical}. Let $D_i$ be the softmax distribution over $[n]$ corresponding to the values $q_i^T k_1,...,q_i^T k_n$. Then it holds that:
\begin{align}
\label{eq:expectation-based-reformulation}
\text{Attn}(q_i,K_i,V_i)= \E\limits_{\ell \sim D_i}[V_{\ell}]
\end{align}
If the stream has length $n$, then it is easy to compute this function exactly for every stream update in $O(nd)$ space by storing the $K_i$ and $V_i$ matrices explicitly. 

\subsection{Johnson-Linderstrauss (JL) Random Projections}
Our analysis makes extensive use of the distributional Johnson-Linderstrauss (JL) lemma (see Appendix \ref{appx:distributional-jl}), which is a fundamental result in the theory of dimensionality reduction. Specifically, we utilize the ability of the standard JL transform to approximately preserve inner products between vectors. The following lemma encapsulates this property:
\begin{lemma}
    \label{lem:inner-product-jl-for-all}
    Suppose we have $n$ points $p_1,...,p_n \in \mathbb{R}^n$ such that $||p_i||_2 \leq 1$ for all $i \in [n]$. Let $f:\mathbb{R}^n\to \mathbb{R}^d$ be a random mapping defined as $f(u) = \frac{1}{\sqrt{d}}Au$ where $A \in \mathbb{R}^{d\times n}$ is a random matrix with its entries drawn independently from a standard normal distribution. Then setting: 
    \begin{align}
        d \geq \frac{12 \ln n}{\varepsilon^2 - \varepsilon^3}
    \end{align}
    allows us to guarantee that with probability at least $1-\frac{1}{n}$ it holds for all $(i,j) \in [n]^2$ that:
    \begin{align}
    |p_i^T p_j - f(p_i)^T f(p_j)| \leq \varepsilon
    \end{align}
\end{lemma}
\begin{proof}
The proof can be found in Appendix \ref{appx:inner-product-jl-for-all}.
\end{proof}

\subsection{Communication Complexity Fundamentals}
In the most basic communication model \citep{rao2020communication}, Alice and Bob possess two $n$-bit strings, $x$ and $y$ respectively. 
They wish to compute a function $f(x, y)$ by exchanging information in a certain number of rounds. 
They both have access to unlimited
computing resources, and their goal is to minimize the total amount of communication, in terms of bits, that they exchange in a fixed number of rounds.
A (potentially randomized) scheme for communicating a function $f$ is called a communication \textbf{protocol}. 
The \textbf{cost} of a protocol is the number of bits exchanged between Alice and Bob in the worst case.
We say that a (randomized) protocol has error $\varepsilon$ if its error probability over its randomness for any input pair
$(x, y)$ is at most $\varepsilon$.

\paragraph{The \textsc{Index} communication problem} The \textsc{Index} problem is a classic communication complexity problem, whose hardness is crucial in our lower bound proofs.

\begin{definition}
In the \textsc{Index} problem, Alice holds a bit string $x \in \{0,1\}^n$ and Bob holds an index $i \in [n]$. Alice sends a single message (one-way) $M \in \{0,1\}^*$ to Bob, whose goal is to output $x_i$ with probability at least $2/3$.
\end{definition}

The hardness of this problem is a well-known result in communication complexity. For completeness, we provide a proof of this result in Appendix \ref{appx:index-lb}.

\begin{theorem}
\label{thm:index-lb}
The one-way, randomized communication complexity of \textsc{Index} is $\Omega(n)$.
\end{theorem}

\subsection{Reservoir Sampling in Streaming}
Reservoir sampling is a widely-used method for sampling in one pass and uniformly at random an element from a data stream in $\widetilde{O}(1)$ memory without knowing the stream's length in advance.
%
%
The $i$-th item is
sampled with probability $\frac{1}{i}$. If it is not sampled, the previously sampled item is retained.
Under this scheme, the probability that any single element is ultimately sampled is $\frac{1}{|\sigma|}$, where $\sigma$ is the length of the stream.

\section{Sublinear Space for KV Cache Compression is Impossible}
\label{sec:lower-bound-main}
In this section, we prove our main lower bound result. We cover the case of exact computation $(\eta = 0)$ on the high-dimensional $(d = \Omega(\log n))$ regime to illustrate the key-ideas of the proof. We later extend this proof with slightly more complicated arguments to the general case. 
\begin{theorem}
    \label{thm:lower-bound-exact}
    There exists some universal constant $C_u > 1$ such that for $d \geq C_u\log n$, any algorithm that can, with probability at least $9/10$ produce an output $o_n 
    \in \mathbb{R}^d$ such that:
    \begin{align}
        o_n = \operatorname{Attn}(q_n,K_n,V_n)
    \end{align}
    must use at least $\Omega(nd)$ bits of memory.
    \end{theorem}
\begin{proof}
    We construct a reduction from \textsc{Index}, in which Alice holds a bit string $x \in \{0,1\}^{n\times d}$ and Bob holds an index $(i,j) \in [n]\times [d]$. Alice sends a single message to Bob, whose goal is to output $x_{ij}$ with probability at least $2/3$. We know by Theorem \ref{thm:index-lb} that the one-way, randomized communication complexity of this problem is $\Omega(nd)$. Our goal is to design a protocol for \textsc{Index} by using a supposed algorithm $\mathcal{A}$ for calculating the attention function that uses $S$ bits of memory. Having that reduction in place, Alice simply communicates these $S$ bits of the algorithm's memory tape to Bob, allowing us to show that $S = \Omega(nd)$, and thus proving the theorem.

    We now describe the protocol for \textsc{Index} using $\mathcal{A}$ and explain the memory lower bound in more detail.

    \paragraph{Alice} Alice begins by inserting the following $n$ triples $\{(q_i,k_i,v_i)\}_{i=1}^n$ of vectors in $\mathbb{R}^d$ to the streaming algorithm $\mathcal{A}$:
    \begin{itemize}
        \item $q_1,...,q_n$ are all the zero vector in $\mathbb{R}^d$, and they do not affect the final output.
         \item $k_1,...k_n \in \mathbb{R}^d$ are calculated before the protocol starts (and agreed upon by both Alice and Bob) as $d$-dimensional projections of the orthonormal basis $e_1 = (1,0,...,0),...,e_n=(0,0,...,1)$ of $\mathbb{R}^n$ in a way that approximately preserves orthonormality\footnote{Equivalently, it suffices to use the JL transform in conjunction with the probabilistic method to show that such vectors exist, as Alice and Bob can pre-compute them before communicating.}. Specifically, we can invoke Lemma \ref{lem:inner-product-jl-for-all} to produce $k_1,...,k_n \in \mathbb{R}^d$ such that with probability at least $1-\frac{1}{n}$ it holds for all $i\neq j$ that:
        $$
        |k_i^T k_j| \leq \varepsilon
        $$
        and for all $i \in [n]$ that:
        $$
        |k_i^T k_i - 1| \leq \varepsilon
        $$
        We do this by letting $f(x) = \frac{1}{\sqrt{d}}Ax$ where $A \in \mathbb{R}^{n\times d}$ is a JL random matrix and defining $k_i = f(e_i)$. Crucially, orthonormality is preserved because $d = \Omega(\log n)$\footnote{If $d \gg \log n$, we can pad the remaining dimensions with zeroes without affecting this construction.}. We resolve the correct value for $\varepsilon$ later in the proof.
        \item $v_1,...,v_n \in \mathbb{R}^d$ contain the rows of $x \in \{0,1\}^{n\times d}$. In other words, Alice uses $V_n$ to store her input $x$ through $\mathcal{A}$:
        \begin{align}
            V_n := x
        \end{align}
    \end{itemize}
    
    After inserting these $n$ triples into $\mathcal{A}$, Alice observes $\mathcal{A}$'s execution and sends its memory state, consisting of $S$ bits, to Bob. This allows Bob to continue the execution of $\mathcal{A}$ exactly where Alice left off, without, of course, having any additional knowledge of $x$.
    
    \paragraph{Bob} Recall that Bob's input is an index $(i,j)$ into $x$. In our protocol, Bob only enters a single triple $(q,\vec{0},\vec{0})$ into $\mathcal{A}$. He defines:
    \begin{align}
        q := C\cdot k_i= C\cdot f(e_i) \in \mathbb{R}^d
    \end{align}
    where $C$ is a positive value to be determined later\footnote{This scaling trick was introduced by \citep{Keles2022} in their fine-grained complexity reductions to show quadratic lower bounds for self-attention computation. It is also used implicitly in \citep{yun2020n}.}. Now, we claim that Bob can recover the value of $x_{ij}$ from the output $o_{n+1}$ that $\mathcal{A}$ produces, which is equal to $\text{Attn}(q,K_{n+1},V_{n+1}) \in \mathbb{R}^d$. We have that:
    \begin{align}
    \text{Attn}(q,K_{n+1},V_{n+1}) = \sigma_{n+1}(K_{n+1}\cdot q)^T\cdot V_{n+1}
    \end{align}
    Starting with $s := K_{n+1}\cdot q \in \mathbb{R}^{n+1}$, we know with probability at least $1-\frac{1}{n}$ that this is a vector in $\mathbb{R}^{n+1}$ with the property that $s_\ell$ is close to $0$ for $\ell\neq i$ and $s_i$ is close to $C$. Specifically, it holds with probability at least $1-\frac{1}{n}$ that:
    \begin{align}
        s_\ell \leq C\varepsilon\,\text{ for } \ell \neq i\,\text{ and }s_{i} \geq C(1-\varepsilon)
    \end{align}
    This is also true vacuously for $\ell = n+1$ as $s_{n+1} = 0$ by construction. Now, let $\xi := \sigma_{n+1}(s) \in \mathbb{R}^{n+1}$. We can see that the softmax vector $\xi$ ``spikes'' at index $i$. We can use this spike to read off the value of $x_{i,j}$ via the $V$ matrix. Let us calculate the product $\xi^T \cdot V_{n+1}$. This is a vector in $\mathbb{R}^d$, whose $j$-th entry is:
    \begin{align}
    (\xi^T\cdot V_{n+1})_j = \sum\limits_{\ell=1}^{n} x_{\ell j}\cdot \xi_{\ell} = x_{ij}\xi_i+\sum\limits_{\ell\neq i}x_{\ell j}\xi_\ell
    \end{align}
    as the last row of $V_{n+1}$ is the zero vector. We can now examine two separate cases:
    \begin{itemize}
        \item \textbf{Case 1: }$x_{ij} = 0$. Then we have that:
        \begin{align}
        (\xi^T\cdot V_{n+1})_j &= \sum\limits_{\ell\neq i}x_{\ell j}\xi_\ell = \frac{\sum\limits_{\ell \neq i}x_{\ell{j}}e^{s_\ell}}{e^{s_i}+\sum\limits_{\ell\neq i}e^{s_\ell}} \leq \frac{\sum\limits_{\ell \neq i}e^{s_\ell}}{e^{s_i}+\sum\limits_{\ell\neq i}e^{s_\ell}} 
        \end{align}
        The function $\frac{x}{x+y}$ is maximized when $x$ is maximized and $y$ is minimized, which allows us to bound:
        \begin{align}
        (\xi^T\cdot V_{n+1})_j \leq \frac{ne^{C\varepsilon}}{ne^{C\varepsilon} + e^{C(1-\varepsilon)}}:=\delta
        \end{align}
        \item \textbf{Case 2: }$x_{ij} = 1$. Then we have that:
        \begin{align}
        (\xi^T\cdot V_{n+1})_j &= \xi_i+\sum\limits_{\ell\neq i}x_{\ell j}\xi_\ell = \frac{e^{s_i}+\sum\limits_{\ell \neq i}x_{\ell{j}}e^{s_\ell}}{e^{s_i}+\sum\limits_{\ell\neq i}e^{s_\ell}} \geq \frac{e^{s_i}}{e^{s_i}+\sum\limits_{\ell\neq i}e^{s_\ell}} 
        \end{align}
        Similarly, the function $\frac{x}{x+y}$ is minimized when $x$ is minimized and $y$ is maximized, so we have:
        \begin{align}
            (\xi^T\cdot V_{n+1})_j \geq \frac{e^{C(1-\varepsilon)}}{e^{C(1-\varepsilon)}+ne^{C\varepsilon}}:=\Delta
        \end{align}
    \end{itemize}
    For Bob to always be able to distinguish between the two cases, we want to ensure that $\delta < \Delta$.
    \begin{align}
    \frac{ne^{C\varepsilon}}{ne^{C\varepsilon} + e^{C(1-\varepsilon)}} &<  \frac{e^{C(1-\varepsilon)}}{e^{C(1-\varepsilon)}+ne^{C\varepsilon}} \iff C > \frac{\ln n}{1-2\varepsilon}
    \end{align}
    Thus, we can set $C = \frac{2\ln n}{1-2\varepsilon}$ and $\varepsilon = 0.1$
    to allow Bob to distinguish between $x_{ij} = 1$ and $x_{ij} = 0$ with probability at least $9/10-1/n > 2/3$. Any choice of $\varepsilon \in (0,1/2)$ would work. Overall, we conclude that the number of bits communicated, which equals the memory of $\mathcal{A}$, is bounded below by a factor of $nd$, which concludes our proof.
\end{proof}

\begin{remark}
    Note that in our reduction Bob can actually recover the entire row $x_{i,:}$. This does not imply a better lower bound because the communication complexity is still $\Omega(nd)$.
\end{remark}

\begin{figure}
    \centering

    \begin{tikzpicture}
                \matrix (K) [xshift=-8cm, matrix of nodes,
            nodes={draw, minimum size=0.5cm, anchor=center},
            column sep=-\pgflinewidth, row sep=-\pgflinewidth
        ] {
            |[fill=gray!20 ]| & |[fill=gray!20 ]| & |[fill=gray!20 ]| & |[fill=gray!20 ]| & |[fill=gray!20 ]| & |[fill=gray!20 ]| & |[fill=gray!20 ]| \\
            |[fill=gray!20 ]| & |[fill=gray!20 ]| & |[fill=gray!20 ]| & |[fill=gray!20 ]| & |[fill=gray!20 ]| & |[fill=gray!20 ]| & |[fill=gray!20 ]| \\
            |[fill=orange!40]| & |[fill=orange!40]| & |[fill=orange!40]| & |[fill=orange!40]| & |[fill=orange!40]| & |[fill=orange!40]| & |[fill=orange!40]| \\
            |[fill=gray!20 ]| & |[fill=gray!20 ]| & |[fill=gray!20 ]| & |[fill=gray!20 ]| & |[fill=gray!20 ]| & |[fill=gray!20 ]| & |[fill=gray!20 ]| \\
            |[fill=gray!20 ]| & |[fill=gray!20 ]| & |[fill=gray!20 ]| & |[fill=gray!20 ]| & |[fill=gray!20 ]| & |[fill=gray!20 ]| & |[fill=gray!20 ]| \\
            |[fill=gray!20 ]| & |[fill=gray!20 ]| & |[fill=gray!20 ]| & |[fill=gray!20 ]| & |[fill=gray!20 ]| & |[fill=gray!20 ]| & |[fill=gray!20 ]| \\
            |[fill=gray!40]|0 & |[fill=gray!40 ]|0 & |[fill=gray!40 ]|0 & |[fill=gray!40 ]|0 & |[fill=gray!40 ]|0 & |[fill=gray!40 ]|0 & |[fill=gray!40 ]|0 \\
        };

        \foreach \i in {1,...,3}
            \node[left=0.2cm of K-\i-1] (label\i) {\footnotesize $k_{\i}$};

        \foreach \i in {4,5}
            \node[left=0.2cm of K-\i-1] (label\i) {\footnotesize $\cdot$};

        \foreach \i in {6}
            \node[left=0.2cm of K-\i-1] (label\i) {\footnotesize $k_{n}$};

        \foreach \i in {7}
            \node[left=0.2cm of K-\i-1] (label\i) {\footnotesize $k_{n+1}$};

        \matrix (Q) [matrix of nodes,
            nodes={draw, minimum size=0.5cm, anchor=center, fill=orange!40},
            column sep=-\pgflinewidth, row sep=-\pgflinewidth,
            right=0.8cm of K
        ] {
            |[fill=orange!40]| \\
            |[fill=orange!40]| \\
            |[fill=orange!40]| \\
            |[fill=orange!40]| \\
            |[fill=orange!40]| \\
            |[fill=orange!40]| \\
            |[fill=orange!40]| \\
        };

         \node[left=0.3cm of Q-4-1]  {$\times$};

        \node[above=0.2cm of Q-1-1] {\footnotesize $q_i \in \mathbb{R}^d$};

        \draw[thick,->] (Q.east) ++(0.3,0) -- ++(0.8,0) node[midway, above] {\footnotesize softmax};
        \begin{scope}[shift={(-3,0)}] 
        \begin{axis}[
            hide axis, 
            width=7.5cm,
            height=3.5cm
        ]
            \addplot[smooth,thick,blue] coordinates {(1,0.05) (2,0.05) (2.4,0.2) (2.6,0.5) (2.95,1) (3.35, 1)(3.65,0.5) (3.85,0.2) (4.15,0.05) (5,0.05) (6,0.05) (7,0.05)};
        \end{axis}
        \end{scope}

        \matrix (m) [yshift=-1cm, matrix of nodes,
            nodes={draw, minimum size=0.7cm, anchor=center},
            column sep=-\pgflinewidth/2, row sep=-\pgflinewidth/2]
        {
            |[fill=white ]| $x_1$ & |[fill=white ]| $x_2$ & |[fill=orange!40 ]| $x_3$ & |[fill=white ]| $\cdot$ & |[fill=white ]| $\cdot$ & |[fill=white ]| $\cdot$ & |[fill=white ]| $x_n$ \\
        };
    \end{tikzpicture}
    \caption{An illustration of the proof of Theorem \ref{thm:lower-bound-exact}. We construct vectors $f(k_1),...,f(k_n)$ such they are approximately orthogonal using the JL transform. Bob holds an index $i$ (in this case $i=3$) and inserts $q_i = f(k_i)$. In this way, the softmax distribution becomes a ``Dirac''-like spike, allowing Bob to read-off Alice's $x_i$ bit (or vector of bits). Note that in this hard instance, the attention matrix is very sparse.}
\end{figure}
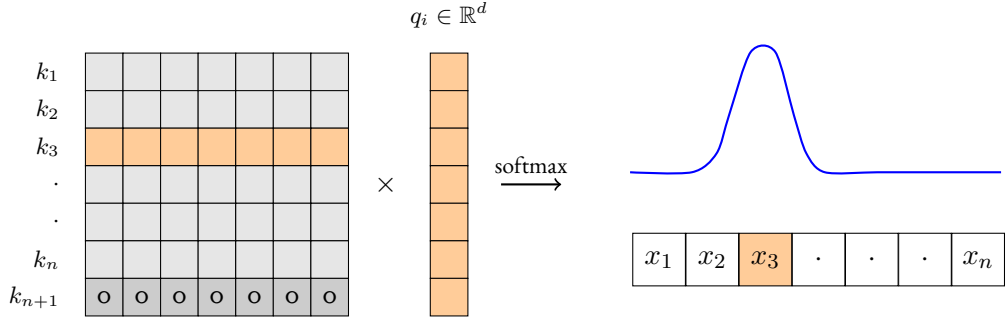

\subsection{Lower bound on approximation algorithms}
Now we extend the above result to the approximate computation of the attention function. 

\begin{theorem}
\label{thm:lower-bound-appx}
Let $Z_n := \operatorname{Attn}(q_n,K_n,V_n)$ and $d = \Omega(\log n)$. Any algorithm that can, with probability at least $9/10$ produce an output $\mathcal{O} \in \mathbb{R}^d$ that is a $(1\pm\eta)$-approximation of $Z_n$
for $\eta \in (0,1)$ must use at least $\Omega(nd)$ bits of memory.
\end{theorem}
\begin{proof}
Our reduction is the same as before, except Bob now relies on the value of $\mathcal{O}$ to distinguish between the two cases $x_{ij}=0$ and $x_{ij} = 1$. When $x_{ij} = 0$ then $\mathcal{O}_j \leq (1+\eta)\delta$, but if $x_{ij} = 1$, then $\mathcal{O}_j \geq (1-\eta)\Delta$,
where $\delta$ and $\Delta$ are defined as in the proof of Theorem \ref{thm:lower-bound-exact}. Then we have to guarantee that $(1+\eta)\delta < (1-\eta)\Delta$ or $\delta < \frac{1-\eta}{1+\eta}\Delta $, which resolves to $C \approx \Omega\left(\ln n - \ln\frac{1-\eta}{1+\eta}\right)$, compensating for $\varepsilon =0.1$
\end{proof}
The result also similarly extends to additive error approximation. The proof is similar to the one above, and we omit it for brevity. 

\subsection{The low-dimensional regime} 
If $d = o(\log n)$, our preceeding proof breaks down because the JL projection is not be able to preserve the inner products of the all pairs of key vectors with high probability. Our technique, however, still yields the following lower bound:
\begin{corollary}
\label{cor:lower-bound-low-dim}
Let $Z_n :=\operatorname{Attn}(q_n,K_n,V_n)$ and $d \geq 2$ with $d = o(\log n)$. Any algorithm that can, with probability at least $9/10$ produce a $(1\pm\eta)$-approximation $\mathcal{O} \in \mathbb{R}^d$ of $Z_n$
for $\eta \in (0,1)$ must use at least $\Omega(e^d\cdot d)$ bits of memory.
\end{corollary}

Surprisingly, this bound is tight for the low-dimensional regime. We show that KV cache compression algorithms that use space sublinear in the number of tokens $n$ do exist in this case. In fact, such an algorithm has already been proposed by the work of \citet{Zandieh2024}. Their algorithm works by assuming a clusterability structure in the key embeddings, defined as follows:
\begin{definition}[Clusterability]
For a positive integer $m$ and real-valued $\delta > 0$, a dataset of points $x_1,...,x_n \in \mathbb{R}^d$ is said to be $(m,\delta)$-clusterable if there exists a partition of the dataset into $m$ clusters $C_1,...,C_m$ such that for all $i \in [m]$ and $x,y \in C_i$ it holds that $||x-y||_2 \leq \delta$.
\end{definition}
The \textsc{SubGen} algorithm of \citep{Zandieh2024} uses this clusterability structure to achieve sublinear space complexity. Their main result is the following:
\begin{theorem}[\textsc{SubGen} algorithm of \citep{Zandieh2024}]
\label{thm:subgen}
Suppose the input is a sequence of tokens $\{(q_i,k_i,v_i)\}_{i=1}^n$ in $\mathbb{R}^d$ such that:
\begin{itemize}
    \item The key embeddings $k_i$ are $(m,\delta)$-clusterable.
    \item $||q_i||_2 \leq r$ for all $i \in [n]$.
\end{itemize}
Then there exists an algorithm that uses $O(d\varepsilon^{-2}\cdot (d+me^{2\delta r}\log n))$ bits of space and outputs an estimate $\widehat{O}_n$ such that with probability at least $0.99$ it holds that:
\begin{align}
    \label{eq:subgen-appx}
    ||\widehat{O}_n - \operatorname{Attn}(q_n,K_n,V_n)||_2 \leq \varepsilon\cdot ||\operatorname{softmax}(K_n\cdot q_n)||_2\cdot ||V_n||_2
\end{align}
\end{theorem}
Our analysis combines this guarantee of \textsc{SubGen} with a lemma on the clusterability of low-dimensional spaces that is based on covering number bounds.
\begin{definition}[Covering Number, \citep{wu2016}]
Define the covering number $N(\Theta, ||\cdot||, \delta)$ of a set $\Theta$ with respect to a norm $||\cdot||$ and a radius $\delta$ as the minimum number of $||\cdot||$-balls of radius $\delta$ needed to cover $\Theta$.
\end{definition}
The following lemma gives an upper bound on the covering number of the unit sphere in $\mathbb{R}^d$.
\begin{lemma}[Covering Number of the Unit Sphere, \citep{wu2016}]
\label{lem:covering-number}
Let $B_2(1)$ be the unit sphere in $\mathbb{R}^d$ with respect to the $l_2$ norm. Then the covering number of $B_2(1)$ with respect to the $l_2$ norm and radius $\delta < 1$ is bounded by:
\begin{align}
N(B_2(1), ||\cdot||_2, \delta) \leq \left(\frac{3}{\delta}\right)^d
\end{align}
\end{lemma}
We can now show that the \textsc{SubGen} algorithm is essentially optimal in the low-dimensional case, modulo some mild assumptions on the norms of the key and query embeddings. 
\begin{lemma}
\label{lem:subgen-clusterability}
Let $x_1,..,x_n \in \mathbb{R}^d$ be a set of $n$ points in $d$-dimensional space such that $||x_i||_2 \leq 1$ for all $i \in [n]$. Then this set is $(m,\delta)$ clusterable for $m = \lceil e^d \rceil$ and $\delta = e/3$.
\end{lemma}
\begin{proof}
Lemma \ref{lem:covering-number} states that we can cover the unit-ball with at most $\left(\frac{3}{\delta}\right)^d = e^d$ balls of radius $\delta$. We then partition the unit-ball into $m $ clusters by assigning each point to a ball that contains it.
\end{proof}
Our penultimate theorem for the low-dimensional case can now be proven:
\begin{theorem}
\label{thm:subgen-optimal}
Let $2 \leq d = o(\log n)$, $\delta := e/3$ and $r := O(\log\log n)$. Suppose it holds for all key embeddings that $||k_i||_2 \leq 1$, for all $i \in [n]$. Then the \textsc{SubGen} Algorithm of Theorem \ref{thm:subgen} approximates the attention function as in Equation \ref{eq:subgen-appx} with space complexity $\widetilde{O}(de^d)$.
\end{theorem}
\begin{proof}
By Lemma \ref{lem:subgen-clusterability}, the key embeddings are $(e^d, e/3)$-clusterable, and so we can apply Theorem \ref{thm:subgen} to get that the space complexity of the \textsc{SubGen} algorithm is:
\begin{align}
    O(d\varepsilon^{-2}\cdot (d+e^{d+2\frac{e}{3} \log\log n}\log n)) = \widetilde{O}(de^d)
\end{align}
The approximation guarantees are as in Equation \ref{eq:subgen-appx}, although with some algebra could be extended to arbitrary absolute error approximation. 
Finally, combining this result with Corollary \ref{cor:lower-bound-low-dim}, we can see that the \textsc{SubGen} algorithm is essentially optimal in the low-dimensional case.
\end{proof}


Finally, for the case of $d=1$, where the key and value vectors are scalars, we can achieve sublinear space utilization using a simpler algorithm. We propose an algorithm that does so in expected $O(\sqrt{n})$ space by utilizing the expectation-based reformulation of attention (Equation \ref{eq:expectation-based-reformulation}) and the lazy Gumbel sampling technique from \citep{mussmann2017fast}. The details of this algorithm are provided in Appendix \ref{appx:alg-low-space}.

\section{How can structural assumptions help? The case of unstructured sparsity}
\label{sec:studying-sparsity}
Achieving sublinear space for KV cache compression, as shown in Theorem \ref{thm:lower-bound-appx}, requires structural assumptions about KV embeddings. A common assumption is the sparsity of the attention matrix $Q^T K$, supported by empirical observations \citep{liu2024scissorhands,liu2023deja,xiaoefficient,Zhang2023}. However, it remains unclear if sparsity alone suffices or if additional structural assumptions are needed. 
In the proof of Theorem \ref{thm:lower-bound-exact}, $Q^T K$ is indeed extremely sparse, with only one element in its last row having significant weight. 
This means that the sparsity of $Q^T K$ by itself is not a strong enough assumption for sublinear space. 
In other words, we have the following corollary:
\begin{corollary}
\label{cor:sparsity-not-enough}
Suppose that an algorithm operates under the promise that a $0.99$-fraction of the entries of $QK^T$ are close to zero. Note that the locations of the non-zero elements are not disclosed. Then, any algorithm that can produce a good approximation of the attention function in this modified streaming setting must use $\Omega(nd)$ bits of memory.
\end{corollary}

Thus, we have to complement sparsity with yet another structural assumption, which is the route taken by many recent works in the literature. 
Specifically, they assume knowledge of the \textit{sparsity pattern} (the \textit{locations} of the non-zero elements) of $Q^T K$ to reduce the memory requirements.

\paragraph{Sliding window attention revisited} The most common such sparsity pattern is the \textbf{sliding window} attention mechanism. In sliding window attention, we do not compare $q_n$ with every prior key vector $k_1,...,k_n$ but only with a window of $W$ prior key vectors: $k_{n-W+1},k_{n-W+2},...,k_n$. Let $S_n := K_n \cdot q_n \in \mathbb{R}^n$ and $h(S_n) := S_n \circ M$ where $\circ$ is element-wise multiplication, $M \in \mathbb{R}^{n}$ is a masking vector such that $M_i = \mathbbm{1}\{i>n-W\}$ and $W$ is the sliding window width. Then we define sliding window attention as:
\begin{align}
    \text{Attn}_{W}(q_n,K_n,V_n):=\sigma(h(S_n))\cdot V
\end{align}
Intuitively, a benefit of this definition is that even though we only attend to $W$ key vectors, we still have access to all value vectors, meaning that every token embedding is considered in the output. 
%

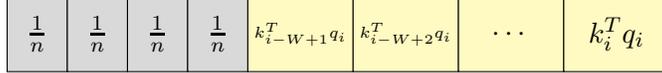
\begin{figure}
\centering
\begin{tikzpicture}
    \def\narrow{0.8} 
    \def\wide{1.4}   

    \fill[gray!30] (0,0) rectangle (\narrow,1);
    \draw (0,0) rectangle (\narrow,1);
    \node at (0.5*\narrow,0.5) {$\frac{1}{n}$};

    \fill[gray!30] (\narrow,0) rectangle (2*\narrow,1);
    \draw (\narrow,0) rectangle (2*\narrow,1);
    \node at (1.5*\narrow,0.5) { $\frac{1}{n}$};

    \fill[gray!30] (2*\narrow,0) rectangle (3*\narrow,1);
    \draw (2*\narrow,0) rectangle (3*\narrow,1);
    \node at (2.5*\narrow,0.5) { $\frac{1}{n}$};

    \fill[gray!30] (3*\narrow,0) rectangle (4*\narrow,1);
    \draw (3*\narrow,0) rectangle (4*\narrow,1);
    \node at (3.5*\narrow,0.5) { $\frac{1}{n}$};

    \fill[yellow!30] (4*\narrow,0) rectangle (4*\narrow+\wide,1);
    \draw (4*\narrow,0) rectangle (4*\narrow+\wide,1);
    \node at (4*\narrow+\wide/2,0.5) {\tiny $k_{i-W+1}^T q_i$};

    \fill[yellow!30] (4*\narrow+\wide,0) rectangle (4*\narrow+2*\wide,1);
    \draw (4*\narrow+\wide,0) rectangle (4*\narrow+2*\wide,1);
    \node at (4*\narrow+1.5*\wide,0.5) {\tiny $k_{i-W+2}^T q_i$};

    \fill[yellow!30] (4*\narrow+2*\wide,0) rectangle (4*\narrow+3*\wide,1);
    \draw (4*\narrow+2*\wide,0) rectangle (4*\narrow+3*\wide,1);
    \node at (4*\narrow+2.5*\wide,0.5) {$\cdots$};

    \fill[yellow!30] (4*\narrow+3*\wide,0) rectangle (4*\narrow+4*\wide,1);
    \draw (4*\narrow+3*\wide,0) rectangle (4*\narrow+4*\wide,1);
    \node at (4*\narrow+3.5*\wide,0.5) {$k_i^T q_i$};

\end{tikzpicture}
    \caption{An illustration of sliding window attention: The softmax entries in $[i-W]$ are all in the support with weight $\approx 1/n$, allowing every value vector to contribute to the final output.}
\end{figure}

However, with this redefinition of sliding window attention, it is no longer clear whether we can calculate it in sublinear space. 
Since our softmax distribution is supported on every index in $[n]$, every row vector of $V$ has a non-zero probability of influencing the output, implying that we might be required to keep all of $V$ in memory. It turns out we can indeed achieve a space complexity of $O(dW)$ even with this definition of sliding window attention.

\paragraph{The Algorithm}
\label{sec:algorithm}
Our algorithm is shown as Algorithm \ref{alg:sliding-window}. 
For the first $W$ token embeddings, we calculate attention directly, while also storing the key and value vectors. 
For the remaining embeddings, we cache the last $W$ key and value vectors, while also maintaining a uniformly sampled value-vector $v_s$ from \textit{outside} the sliding window. 
We can achieve this using reservoir sampling (Section \ref{sec:prelims}).
We sample from the sliding window softmax distribution by examining two cases. 
In the first case, we sample an index inside the sliding window using our explicitly built cache.
In the second case we sample an index outside the sliding window by using our reservoir sample.
As a result we obtain an unbiased estimator of the sliding window attention function, which we can boost using the median-of-means technique. The analysis can be found in Appendix \ref{appx:sliding-window-analysis}.
 
\begin{algorithm}[ht]
\caption{Sliding Window Attention in Sublinear Space}\label{alg:sliding-window}
\begin{algorithmic}[1]
\vspace{1mm}
\STATE \textbf{Inputs}: Window size $W > 0$.
\vspace{1mm}
\STATE Let $K,V \gets \emptyset,\emptyset$ be the list of key-vectors and value-vectors to store.
\FOR{$i = 1$ to $W$}
    \STATE Receive token embeddings $(q_i,k_i,v_i) \in (\mathbb{R}^d)^3$
    \STATE $K\gets \text{append }(k_i)$, and $V\gets \text{append }(v_i)$.
    \STATE \textbf{Output} $\text{Attn}_{W}(q_i,K,V)$ using direct calculation.
\ENDFOR
\vspace{1mm}
\STATE Let $v_s \gets \emptyset$ store a randomly selected value-vector outside the sliding window.
\vspace{1mm}
\FOR{$i = W+1$ to $n$}
    \STATE Receive token embeddings $(q_i,k_i,v_i) \in (\mathbb{R}^d)^3$
    \STATE $v \gets$ Pop the first element of $V$ and remove the first element of $K$.
    \STATE $K\gets \text{append }(k_i)$, and $V\gets \text{append }(v_i)$.
    \vspace{1.5mm}
    \STATE With probability $(i-W)^{-1}\to$ replace: $v_s \gets v$.
    \vspace{1mm}
    \STATE Let $S := (i-W)+\sum_{k \in K}e^{q_i^T k}$
    \vspace{1mm}
    \STATE With probability $(S-i+W)/S$:
    \vspace{1mm}
    \STATE \hspace{1em} Sample $\ell \in \{i-W+1,...,i\}$ wp $\propto e^{q_i^T k_\ell}$ and \textbf{Output} $V_{\ell} \in \mathbb{R}^d$.
    \STATE Else:
    \STATE \hspace{1em} \textbf{Output} $v_s \in \mathbb{R}^d$.
\ENDFOR
\end{algorithmic}
\end{algorithm}




\subsection{A Tight Lower Bound for Sliding Window Attention}
We can show that the space complexity of $O(dW)$ is tight for sliding window attention, even with our updated definition. We include the proof in Appendix \ref{appx:lower-bound-sliding-window}.

\begin{theorem}
    \label{thm:lower-bound-sliding-window}
    For $d = \Omega(\log n)$, any algorithm that can, with probability at least $9/10$ approximate sliding window attention must use at least $\Omega(dW)$ bits of memory.
\end{theorem}

\section{Non-Adaptive Lower Bounds for the Time Complexity of Token Generation}
\label{sec:time-complexity}
\begin{definition}
We call an algorithm \textit{non-adaptive} if the data accesses / queries it performs are pre-decided (either deterministically or through a randomized procedure) before the data arrives. 
\end{definition}
\begin{theorem}
\label{thm:time-complexity}
Let $Z_n := \operatorname{Attn}(q_n,K_n,V_n)$ and suppose $\mathcal{A}$ is a non-adaptive streaming algorithm for token generation that calculates the attention function in an online manner and outputs a $(1\pm \eta)$-approximation $\mathcal{O} \in \mathbb{R}^d$ of $Z_n$ with probability at least $9/10$ for $\eta \in (0,1)$. Then $\mathcal{A}$ must take $\Omega(nd)$ time to process the last token $(q_n,k_n,v_n)$.
\end{theorem}
\begin{proof}
The full proof can be found in Appendix \ref{appx:time-complexity}, but we provide an outline here. We start from the expectation-based reformulation of self-attention (Equation \ref{eq:expectation-based-reformulation}):
\begin{align}
    \text{Attn}(q_n,K_n,V_n) = \E_{\ell \sim D_n}[V_{\ell}]
\end{align}
We can treat $\mathcal{A}$ as an algorithm that makes queries to $D_n$ to estimate this expectation. Each query takes time $\Theta(d)$ just to calculate the corresponding score $q_n^T k_\ell$, so it suffices to show that $\Omega(n)$ queries are required to obtain a good approximation with high constant probability.

We use the technique of \textit{indistinguishability}. We construct a family of stream instances $\mathcal{H} = H_1,...,H_n$ and an instance $\sigma$ such that the distribution $D_n$ in any $H_i$ differs from the distribution $D_n$ in $\sigma$ in at most one position. Thus, if $\mathcal{A}$ makes $o(n)$ non-adaptive queries to $D_n$, then with probability at least $1/\text{poly}(n)$ it will not be able to distinguish between $\sigma$ and a randomly chosen $H_i$. However, our construction ensures that the attention function differs significantly between $\sigma$ and $H_i$, contradicting the assumption that $\mathcal{A}$ approximates the attention function well.
\end{proof}

\paragraph{What about adaptive algorithms?} Theorem \ref{thm:time-complexity} only applies to non-adaptive algorithms. For adaptive algorithms, the time complexity can be significantly reduced. For instance, a $k$-nearest-neighbors approach can achieve $O(d\sqrt{n})$ time per token generation step. A detailed analysis of this method can be found in the work of \citet{haris2024knnattentiondemystifiedtheoretical}.

\section{Conclusion}
In this paper, we performed a theoretical study of the space and time complexity of approximating the attention function in Transformer LLMs during inference. We demonstrated that linear memory is essentially required, implying that efficient KV cache compression must make structural assumptions about the attention function inputs. Several future directions arise in the aftermath of our investigation, such as, but not limited to, the following:
\begin{enumerate}[itemsep=2pt]
    \item Does sparsity or structure in general arise naturally when learning Transformer models? This seems to happen in practice, but a theoretical study of such phenomena is still missing.
    \item Can sublinear space be achieved in the average case, under distributional assumptions on the input embeddings?
\end{enumerate}
Overall, our work contributes to the growing theoretical analysis of deep learning models, which, despite being relatively poorly understood, are widely used in modern applications.

\section*{Acknowledgments}
We thank Fabian Spaeh and Rathin Desai for their comments on an earlier draft of this paper. 

\bibliographystyle{plainnat}
\bibliography{references}

\appendix

\section{Extended Review of Related Work}
\subsection{KV Cache Compression}

A significant body \citep{shi2024keep} of recent research has focused on compressing the key-value (KV) cache in Transformer-based large language models (LLMs), aiming to reduce memory footprint during inference. These algorithms typically exploit either inherent structural sparsity or learned sparsity patterns within the input embeddings.

\citet{Zhang2023} introduced $H_2 O$, an eviction strategy for the KV cache that capitalizes on the sparse nature of the attention matrix. Their approach leverages techniques from dynamic submodular maximization. \citet{xiaoefficient} made the empirical observation that initial tokens in an input sequence, which they termed "attention sinks," often maintain high attention weights throughout the generation process. This observation suggests a natural compression strategy: storing these attention sinks along with a sliding window of recent token embeddings. Similarly, \citet{liu2023scissorhands} observed that tokens with high attention weights tend to retain their influence on the attention output—a phenomenon they call ``persistence of importance". The work of \citet{liu2023deja} emphasized the importance of "contextual sparsity" in Transformer models.  \citet{Ge2023} proposed an adaptive approach to KV cache construction, dynamically maintaining the most important tokens encountered.  Other dynamic algorithms for KV cache compression can be found in \citep{cai2024pyramidkv,zhang2024unifying,tang2024razorattention}.

Beyond sparsity-based assumptions, \citet{Zandieh2024} leveraged ``clusterability'' structure of key embeddings to achieve provably sublinear space complexity. Hardware-oriented methods have also been explored. For instance, \citet{Duanmu2024} and \citet{zandieh2024qjl} utilized quantization techniques, while \citet{Jin2024} examined the input/output (I/O) complexity of the problem. \citet{liu2024cachegen} tackled the challenge of efficiently serving LLMs during inference by optimizing the KV cache to match available network bandwidth.  Taking a broader perspective, \citet{liu2024minicache} proposed a layer-specific strategy, improving the relationship between compression ratio and embedding dimension.

For a more comprehensive review of KV cache compression techniques, we refer the reader to \citep{li2024survey} and \citep{shi2024keep}.

\subsection{Sparse Attention Mechanisms}
Improving the time and space complexity of self-attention has been thoroughly studied beyond the context of serving LLMs at inference time, because it poses a significant bottleneck in the training and fine-tuning stages of Transformer models as well. Many algorithms for efficient attention computation have been proposed, with many leveraging inherent sparsity properties \citep{kitaev2020reformer,choromanski2020rethinking,wang2020linformer,haris2024knnattentiondemystifiedtheoretical,child2019generating,correia2019adaptively}, others utilizing low-rank approximations \citep{han2023hyperattention,cahyawijaya2021greenformers,chen2021scatterbrain,li2024learning,katharopoulos2020transformers}, and many more focusing on optimal exploitation of the available hardware \citep{pagliardini2023faster,dao2022flashattention,dao2023flashattention,shah2024flashattention}. The design of these algorithms undoubtedly influences the research landscape of KV cache compression, except the problem of training is an offline one, while the problem of inference is online.
For a survey of the various techniques for efficient attention computation, we refer the reader to \citep{fournier2023practical}.

\subsection{Impossibility Results for Transformers}
As research towards efficient Transformers progresses, complexity-theoretic techniques have also helped shed light on their limitations. 
The computation of the full attention matrix has been shown to require $\Omega(n^2)$ time via fine-grained complexity reductions \citep{Keles2022,alman2024fast}. 
Similar impossibility results have also been shown for the attention gradients \citep{alman2024fine}. 
The I/O complexity of self-attention has been studied by \citep{saha2024complexity}, who also used communication complexity techniques in their impossibility results. 
Important insights have also been obtained by studying the limitations of Transformers as a computational model through the lens of communication complexity and circuit complexity \citep{alman2024fundamental,chen2024circuit}. 
The representational and generalization power of Transformers has also been theoretically studied \citep{sanford2024transformers,sanford2024representational,li2023transformers,yun2020n}.

\section{Preliminaries on the Distributional JL Lemma}
\label{appx:distributional-jl}
We use the following classic distributional version of the Johnson-Linderstrauss (JL) norm-preservation theorem in our proofs:
\begin{theorem}
    \label{jl-thm-norms}
    Let $x \in \mathbb{R}^n$. Let $A \in \mathbb{R}^{n\times d}$ be a matrix whose entries are i.i.d. standard Gaussian random variables: $A \sim \mathcal{N}(0,1)^{n\times d}$. Then, it holds that:
    \begin{align}
    \Pr_{A}\left[\left|\left|\frac{1}{\sqrt{d}}Ax\right|\right|_2^2 \notin (1\pm \varepsilon)||x||_2^2\right] \leq 2e^{-\frac{(\varepsilon^2-\varepsilon^3)d}{4}}
    \end{align}
\end{theorem}
A simple modification of this theorem allows us to show that inner products are preserved as well.
\begin{lemma}
    \label{lem:inner-product-jl-lemma}
    Let $x,y \in \mathbb{R}^n$ with $||x||_2,||y||_2 \leq 1$. As before, let $f:\mathbb{R}^n\to\mathbb{R}^d$ where $f(u) = \frac{1}{\sqrt{d}}Au$ and $A \in \mathbb{R}^{n\times d}$ is a matrix whose entries are i.i.d. standard Gaussian random variables. Then, it holds that:
    \begin{align}
    \Pr_A\left[\left|x^T y - f(x)^T f(y)\right| \geq \varepsilon\right] \leq 4e^{-\frac{(\varepsilon^2-\varepsilon^3)d}{4}}
    \end{align}
\end{lemma}
\begin{proof}
    We can apply Theorem \ref{jl-thm-norms} to vectors $x-y$ and $x+y$ to get that with probability at least $1-4e^{-\frac{(\varepsilon^2-\varepsilon^3)d}{4}}$ it holds that:
    \begin{align}
    \left|\left|f(x-y)\right|\right|_2^2 \in (1\pm \varepsilon)||x-y||_2^2\\
    \left|\left|f(x+y)\right|\right|_2^2 \in (1\pm \varepsilon)||x+y||_2^2
    \end{align}
    Now, by the linearity of $f$, the polarization identity and the parallelogram law, we have that:
    \begin{align}
    f(x)^T f(y) &= \frac{1}{4}\left(||f(x)+f(y)||_2^2 - ||f(x)-f(y)||_2^2\right)\\
     &= \frac{1}{4}\left(||f(x+y)||_2^2 - ||f(x-y)||_2^2\right)\\
    &\geq\frac{1}{4}\left( (1-\varepsilon)||x+y||_2^2 - (1+\varepsilon)||x-y||_2^2\right)\\
    &=\frac{1}{4}\left(4x^T y - 2\varepsilon(||x||_2^2+||y||_2^2)\right)\\
    &\geq\frac{1}{4}\left(4x^T y - 4\varepsilon\right)\\
    &= x^T y - \varepsilon
    \end{align}
    Similarly we can prove that $f(x)^T f(y) \leq x^T y + \varepsilon$.
\end{proof}

\section{Proof of Lemma \ref{lem:inner-product-jl-for-all}}
\label{appx:inner-product-jl-for-all}

\begin{lemma}
    Suppose we have $n$ points $p_1,...,p_n \in \mathbb{R}^n$ such that $||p_i||_2 \leq 1$ for all $i \in [n]$. Let $f:\mathbb{R}^n\to \mathbb{R}^d$ be a random mapping defined as in Lemma \ref{lem:inner-product-jl-lemma}. Then, setting
    \begin{align}
        d \geq \frac{12 \ln n}{\varepsilon^2 - \varepsilon^3}
    \end{align}
    allows us to guarantee that with probability at least $1-\frac{1}{n}$ it holds that:
    \begin{align}
    |p_i^T p_j - f(p_i)^T f(p_j)| \leq \varepsilon
    \end{align}
    for all $(i,j) \in [n]^2$.
\end{lemma}
\begin{proof}
    We need to set $d$ such that for any given fixed pair $p_i,p_j$ it holds that:
    $$
        \Pr_f\left[|x^T y - f(x)^T f(y)| \geq \varepsilon\right] \leq e^{-\frac{(\varepsilon^2 - \varepsilon^3)k}{4}} \leq \frac{1}{n^3}
    $$
    If we do this, then by union-bound over all $n^2$ pairs we get the desired property. So, we solve for $d$ in the following inequality:
    \begin{align}
    &\exp\left(-\frac{(\varepsilon^2 - \varepsilon^3)d}{4}\right) \leq \frac{1}{n^3} \\
    &\iff -\frac{(\varepsilon^2 - \varepsilon^3)d}{4} \leq -\ln n^3\\
    &\iff d \geq \frac{12 \ln n}{\varepsilon^2 - \varepsilon^3}
    \end{align}
    This establishes the value of $d$ and completes the proof.
\end{proof}    

\section{Proof of Theorem \ref{thm:index-lb}}
\label{appx:index-lb}
\begin{theorem}
The one-way, randomized communication complexity of the \textsc{Index} problem is $\Omega(n)$.
\end{theorem}
\begin{proof}
We give the proof found in the lecture notes of David Woodruff's 2020 course on Algorithms for Big Data \citep{woodruff2020}. Let $M$ be the message that Alice sends to Bob (as a random variable). Consider the uniform distribution over $\{0,1\}^n$ and let $X$ be Alice's input, drawn from this uniform distribution. Then Bob outputs an estimate $X_j'$ of $X_j$ for his input index $j$. Since $X_j'$ only depends on $M$, $X_j\to M \to X_j'$ forms a Markov chain. We can thus use Fano's inequality:
\begin{align}
    H(X_j\mid M) \leq H_b(X_j \neq X_j') + \Pr[X_j \neq X_j']\log_2(2-1) \leq H_b(1/3)
\end{align}
Now we use this to bound the mutual information between the message and the input $X$. We have by the chain rule of mutual information that:
\begin{align}
    I(X;M) = \sum\limits_{i=1}^n I(X_i;M\mid X_1,...,X_{i-1}) &= \sum\limits_{i=1}^n H(X_i\mid X_1,...,X_{i-1}) - H(X_i\mid M,X_1,...,X_{i-1})\\
    &n-\sum\limits_{i=1}^n H(X_i\mid M)\\
    &\geq n-n\cdot H_b(1/3) \\
    &=\Omega(n)
\end{align}
We are essentially done. If $|M| = \ell$, then $M$ takes values uniformly in the space $\{0,1\}^\ell$, so $H(M) \leq \log_2(2^\ell) = \ell$. Thus, $|M| \geq H(M) \geq I(X;M)$, which implies that $\ell = \Omega(n)$.
\end{proof}

\section{A Sublinear Space Algorithm for the One Dimensional Case}
\label{appx:alg-low-space}
In this section, we provide an algorithm that uses sublinear space to calculate the self-attention function for token generation when $d=1$. We use a technique that leverages both the expectation-based reformulation of self-attention (Equation \ref{eq:expectation-based-reformulation}) and the method of lazy Gumbel sampling from \citep{mussmann2017fast}. Crucially, this technique only works in the degenerate $d=1$ case, for which it yields sublinear space, while for larger $d$ linear space can be required. 

\subsection{Preliminaries}
We shall make use of the \textit{lazy Gumbel sampling} method, proposed by \citet{mussmann2017fast}, for sampling from a softmax distribution. Suppose we are given a set of scores $s_1,...,s_n$ and we wish to sample from the distribution $p_i = \frac{e^{s_i}}{\sum_{j=1}^n e^{s_j}}$. The well-known \textit{Gumbel-max trick} allows us to sample from this distribution by adding Gumbel noise to the scores and taking the argmax:
\begin{align}
    \text{Sample } i \sim p \iff i = \arg\max_{j\in [n]}(s_j + g_j)\,\text{ where } g_j \sim \text{Gumbel}(0,1)
\end{align}
Recall that the Gumbel distribution is defined as:
\begin{align}
    \text{Gumbel}(0,1) = -\log(-\log(u))\,\text{ where } u \sim U(0,1)
\end{align}
In the lazy Gumbel sampling method, we only need to sample Gumbel noise for the top $\sqrt{n}$ scores. Because the Gumbel distribution has a light tail, the remaining scores are very unlikely to affect the argmax, so it suffices to sample around $\sqrt{n}$ of them in expectation. Overall, just by adding Gumbel noise to $\Theta(\sqrt{n})$ scores in expectation, we can sample from $p$. The algorithm, adapted from \citet{mussmann2017fast}, is given in Algorithm \ref{alg:lazy-gumbel} for completeness (adapted from the work of \citet{haris2024knnattentiondemystifiedtheoretical})

\begin{algorithm}
    \begin{algorithmic}[1]
    \caption{Lazy Gumbel Sampling from $D_i$, for some $i \in [n]$, \citep{mussmann2017fast}}
    \label{alg:lazy-gumbel}
        \STATE \textbf{Inputs:} $k\in\mathbb{N}$, $q_i \in \mathbb{R}^d, K \in 
        \mathbb{R}^{n\times d}$, $S_i := \{ \text{the $k$ keys $j$ 
        with the largest $Z_{ij} := q_i^T k_j$}\}$.
        \vspace{1mm}
        \STATE Sample $G_{ij} \sim $ Gumbel$(0,1)$ for $j \in S_i$.
        \vspace{1mm}
        \STATE Let $M \gets \max\limits_{j \in S_i} \{Z_{ij} + G_{ij}\}$ and $S_{\min} \gets \min\limits_{j \in S_i} \{Z_{ij}\}$.
        \STATE Let $B \gets M - S_{\min}$ be the Gumbel cutoff.
        \vspace{1mm}
        \STATE Let $m \sim \text{Bin}(n-k, 1-\exp(-\exp(-B)))$ be the number of $
        [n]\setminus S_i$ Gumbels greater than $B$. Sample $m$ points from $[n]\setminus 
        S_i$ and denote the set of sampled points as $T_i$.
        \vspace{1mm}
        \STATE Sample $G_{ij} \sim \text{Gumbel}(0,1)$ conditionally greater than $B$ for each $j \in T_i$.
        \vspace{1mm}
        \STATE Return $\widehat{j} \gets \arg\max\limits_{j \in S_i\cup T_i} \{Z_{ij} + 
        G_{ij}\}$
    \end{algorithmic}
\end{algorithm}

\subsection{Algorithm}
We provide an outline of the algorithm for the $d=1$ case, leveraging the lazy Gumbel sampling method. Recalling the expectation-based reformulation of self-attention (Equation \ref{eq:expectation-based-reformulation}), we can write the self-attention function in update $(q,k_n,v_n)$ as:
\begin{align}
    \text{Attn}(q,K_n,V_n) = \E_{\ell \sim D}[V_{\ell}]
\end{align}
where $D$ is a softmax distribution defined by the scores $q^T k_1,...,q^T k_n$. We can sample from this distribution using Algorithm \ref{alg:lazy-gumbel}, but we need to be careful about two caveats:
\begin{itemize}
    \item For each update in the input stream we need to know the top $\sqrt{n}$ scores $q^T k_j$ for $j \in [n]$. Since $d=1$, $q,k_n,v_n \in \mathbb{R}$, so we can simply keep a buffer of the top $\sqrt{n}$ and the bottom $\sqrt{n}$ key scalars (depending on whether $q >0$ or $q<0$) and update them as new keys arrive. We can do this in $O(\sqrt{n})$ space.
    \item We also need to sample $\sqrt{n}$ scores (in expectation) outside of the top scores to add Gumbel noise to. We can do this using reservoir sampling, which requires $O(\sqrt{n})$ space as well. 
\end{itemize}
Since we can sample from $D$, we can form an unbiased estimator of the self-attention function. Our estimator can be boosted to arbitrary precision, akin to Theorem \ref{thm:sliding-window-sublinear-space}. As a result, $O(\sqrt{n})$ space in expectation is sufficient to calculate the self-attention function for token generation in the $d=1$ case.

\section{Analysis of Algorithm \ref{alg:sliding-window}}
\label{appx:sliding-window-analysis}
We now prove the correctness of Algorithm \ref{alg:sliding-window} and analyze its time and space complexity.
\begin{theorem}[Analysis of Algorithm \ref{alg:sliding-window}]
  For each update step $i \in [n]$, Algorithm \ref{alg:sliding-window} produces an unbiased estimator $O_i$ of the sliding window attention function. In other words:
  \begin{align}
      \mathbb{E}[O_{ij}] = \operatorname{Attn}_W(q_i,K_i,V_i)_j
  \end{align}
  for all $j \in [d]$. Furthermore, Algorithm \ref{alg:sliding-window} uses $O(dW)$ space and $O(dW)$ time per update step.
  \end{theorem}
  \begin{proof}
  The time and space complexity is easy to see: We only store at most $W$ vectors in the KV cache, so our space utilization is $O(dW)$.
  For correctness, we start with the expectation-based reformulation of sliding window attention (Section \ref{sec:prelims}).
  Let $D_i$ be the softmax distribution over $[n]$ from the values corresponding to the sliding window: $(0,0,...,0,q_i^T k_{i-W+1},...,q_i^T k_i)$. Then, it is true that
  \begin{align}
  \text{Attn}_W(q_i,K_i,V_i) = \E\limits_{\ell \sim D_i}[V_{\ell}]
  \end{align}
  and so we can, using the law of total expectation, write that:
  \begin{align}
  \text{Attn}&_W(q_i,K_i,V_i)_j =\\[2mm]
  &\E\limits_{\ell \sim D_i}[V_{\ell j}\mid \ell \leq i-W]\cdot \Pr_{\ell \sim D_i}[\ell \leq i-W]+\E\limits_{\ell \sim D_i}[V_{\ell j}\mid \ell > i-W]\cdot \Pr_{\ell \sim D_i}[\ell > i-W]\\
  &=\E\limits_{\ell\sim U(1,i-W)}[V_{\ell j}]\cdot \frac{i-W}{i-W+\sum\limits_{\ell=i-W+1}^i e^{q_i^T k_\ell}}+\E\limits_{\ell\sim D_i'}[V_{\ell j}]\cdot \frac{\sum\limits_{\ell=i-W+1}^i e^{q_i^T k_\ell}}{i-W+\sum\limits_{\ell=i-W+1}^i e^{q_i^T k_\ell}}
  \end{align}
  where $U(\cdot)$ is the uniform distribution and $D_i'$ is a distribution with support the set $\{i-W+1,...,i\}$ defined as:
  \begin{align}
      D_i'(\ell) = \begin{cases}
          0,&\text{if } \ell \leq i-W\\
          \frac{e^{q_i^T k_\ell}}{\sum\limits_{\ell=i-W+1}^i e^{q_i^T k_\ell}},&\text{otherwise}
      \end{cases}
  \end{align}
  On the other hand, if $\mathcal{E}_22$ is the event that Line 22 is executed and $\mathcal{E}_{24}$ is the event that line 24 is executed, then, by construction, we can calculate the expected value $\mathbb{E}[O_{ij}]$ for any $j \in [d]$ as:
  \begin{align}
  &\mathbb{E}[O_{ij}\mid \mathcal{E}_{22}]\cdot \frac{\sum\limits_{\ell=i-W+1}^i e^{q_i^T k_\ell}}{i-W+\sum\limits_{\ell=i-W+1}^i e^{q_i^T k_\ell}}+\mathbb{E}[O_{ij}\mid \mathcal{E}_{24}]\cdot \frac{i-W}{i-W+\sum\limits_{\ell=i-W+1}^i e^{q_i^T k_\ell}}
  \end{align}
  We know that
  \begin{align}
  \mathbb{E}[O_{ij}\mid \text{Line 22 is executed}] = \mathbb{E}_{\ell\sim D_i'}[V_{\ell j}]
  \end{align}
  by construction. 
  Furthermore, $v_s$ is uniformly sampled from the set $\{v_1,...,v_{i-W}\}$ due to the correctness of reservoir sampling, so
  \begin{align}
  \mathbb{E}[O_{ij}\mid \text{Line 24 is executed}]  = \mathbb{E}_{\ell\sim U(1,i-W)}[V_{\ell j}]
  \end{align}
  We can therefore conclude that:
  \begin{align}
      \mathbb{E}[O_{ij}] = \text{Attn}_W(q_i,K_i,V_i)_j
  \end{align}
  for all $j \in [d]$, and that our algorithm produces an unbiased estimator of the attention function.
  \end{proof}

The following theorem analyzes the boosted estimator:

\begin{theorem}
\label{thm:sliding-window-sublinear-space}
Assume that for all $(i,j) \in [n]\times[d]$ it holds that $|\operatorname{Attn}(q_i,K_i,V_i)_j| = \Omega(1)$. Then, there exists an algorithm that uses $O(||V||_\infty\cdot \varepsilon^{-2}\cdot \log(nd/\delta))$
independent samples of the estimator given by Algorithm \ref{alg:sliding-window} to produce an estimator $\widehat{O}_i$ such that with probability at least $1-\delta$ it holds that for all $i \in [n]$ and $j \in [n]$ that:
\begin{align}
    |\widehat{O}_{ij}-\operatorname{Attn}_W(q_i,K_i,V_i)_j| \leq \varepsilon\cdot \operatorname{Attn}_W(q_i,K_i,V_i)_j
\end{align}
Overall, our algorithm has space complexity of 
\begin{align}
O(dW\cdot||V||_\infty\cdot \varepsilon^{-2}\cdot \log(d/\delta))
\end{align}
\end{theorem}
\begin{proof}
First, we calculate the variance of our estimator. Let $j \in [d]$ and let:
\begin{align}
    S&:=i-W+\sum\limits_{\ell=i-W+1}^i e^{q_i^T k_\ell},\,\text{ and }\\
    S_W &:= \sum\limits_{\ell=i-W+1}^i e^{q_i^T k_\ell}
\end{align}
Then we can bound the variance as:
\begin{align}
\text{Var}[O_{ij}]&\leq \mathbb{E}[O_{ij}^2]\\
&=\mathbb{E}[O_{ij}^2\mid \text{Line 23 is executed}]\cdot \frac{S_W}{S}+\mathbb{E}[O_{ij}^2\mid \text{Line 25 is executed}]\cdot \frac{i-W}{S}\\
&\leq \E\limits_{\ell \sim D_i'}[V_{\ell j}^2]\cdot \frac{S_W}{S}+\E\limits_{\ell \sim U(1,i-W)}[V_{\ell j}^2]\cdot \frac{i-W}{S}\\
&=\frac{S_W}{S}\sum\limits_{\ell=i-W+1}^i D_i'(\ell)\cdot V_{\ell j}^2 +\frac{i-W}{S}\sum\limits_{i=1}^{i-W}\frac{1}{i-W}\cdot V_{\ell j}^2 \\
&\leq ||V||_\infty \cdot \mathbb{E}[O_{ij}]
\end{align}

Now, suppose $j\in[d]$ is some index, and we take
\begin{align}
T = \frac{3\text{Var}[O_{ij}]}{\varepsilon^2 \E[O_{ij}]^2}\leq \frac{3||V||_\infty}{\varepsilon^2 \E[O_{ij}]}
\end{align}
samples $o_{1j},...,o_{Tj}$ of our estimator $O_{ij}$ by running Algorithm \ref{alg:sliding-window} in parallel, and let $\widetilde{o}_j$ be their average: $\widetilde{o}_j = \frac{1}{T}\sum\limits_{t=1}^T o_{tj}$. Then, by Chebyshev's inequality, we have that:
\begin{align}
    \Pr[|\widetilde{o}_j-\E[O_{ij}]| \geq \varepsilon\cdot \E[O_{ij}]] &\leq \frac{\text{Var}[O_{ij}]}{T\varepsilon^2\cdot \E[O_{ij}]^2} = \frac{1}{3}
\end{align}

Now, we take $Q = \Theta\left(\log\left(\frac{1}{\delta}\right)\right)$ samples of $\widetilde{o}_j$: $\widetilde{o}_{1j},...,\widetilde{o}_{Qj}$ and take their median $\widehat{O}_{j}$ as our final estimator:
\begin{align}
    \widehat{o}_j := \text{median }(\widetilde{o}_{1j},...,\widetilde{o}_{Qj})
\end{align}
For each $t \in [Q]$, let us define the event $E_t$ as:
\begin{align}
    E_t = \left\{|\widetilde{o}_{tj}-\E[O_{ij}]| \geq \varepsilon\cdot \E[O_{ij}]\right\}
\end{align}
and let $X_t$ be the indicator random variable for $E_t$. Let $X := \sum\limits_{t=1}^Q X_t$ be the total number of events that occur. We know that the $X_t$-s are independent, and that $\mathbb{E}[X_t] = \Pr[E_t] \leq \frac{1}{3}$, which implies that $\mathbb{E}[X] \leq \frac{Q}{3}$. Now, because our estimator is the median, it over-estimates only if more than half of the samples over-estimate. Using a standard Chernoff bound, this implies that:
\begin{align}
    \Pr\left[\widehat{o}_j \geq (1+\varepsilon)\E[O_{ij}]\right] &\leq \Pr\left[X \geq \frac{Q}{2}\right] \\
    &=\Pr\left[X \geq \left(1+\frac{1}{2}\right)\cdot\frac{Q}{3}\right]\\
    &\leq e^{-\frac{Q}{12}}\\
    &\leq \frac{\delta}{2}
\end{align}
given our choice of $Q$. A symmetric argument establishes the under-estimation bound, meaning that $\widetilde{o}_j$ is an $\varepsilon$ multiplicative approximation of $\mathbb{E}[O_{ij}]$ with probability at least $1-\delta$. Since we have assumed that $\E[O_{ij}] = \Omega(1)$ and the total number of samples is $TQ$, we arrive at our space complexity. Finally, we set $\delta :=\delta / (nd)$ and union-bound over all $d$ dimensions and $n$ timesteps to get our desired result. 
\end{proof}

\section{Proof of Theorem \ref{thm:lower-bound-sliding-window}}
\label{appx:lower-bound-sliding-window}
\begin{theorem}
    Let $W_n := \operatorname{Attn}_W(q_n, K_n, V_n)$. For $d = \Omega(\log n)$, any algorithm that can, with probability at least $9/10$ produce an output $o_n 
    \in \mathbb{R}^d$ such that:
    \begin{align}
        |o_{nj}-W_{nj}|\leq \eta\cdot W_{nj}
    \end{align}
    for $\eta \in (0,1)$ and $j \in [d]$ must use at least $\Omega(Wd)$ bits of memory.
\end{theorem}
\begin{proof}
    We use the same idea from the proofs of the main lower bound, which is a reduction from the \textsc{Index} communication problem. Now, Alice holds a bit string $x \in \{0,1\}^{W\times d}$ and Bob holds $(i,j) \in [W]\times [d]$. Let $\mathcal{A}$ be a streaming algorithm for approximating sliding window attention that uses $S$ bits of memory. Alice starts by adding to $\mathcal{A}$ $n$ triples $(q_i,k_i,w_i)$, where:
    \begin{itemize}
        \item $q_i = \vec{0}$,
        \item $k_i = \vec{0}$ for $i \leq n-W$ and $k_i = f(e_i)$ for $i \geq n-W+1$. Here $f(x) := \frac{1}{\sqrt{d}}Ax$ is defined again as the linear JL transform that approximately preserves inner products. 
        \item $v_i = \vec{0}$ for $i \leq n-W$ and $v_i = x_{i,:}$ for $i \geq n-W+1$. In other words, Alice uses the last $W$ rows of $V$ to encode the information in $x$ through $\mathcal{A}$.
    \end{itemize}
    
    Alice communicates to Bob the $S$ bits of the memory tape of $\mathcal{A}$ after it has processed these $n$ timesteps. Bob now enters one more input to $\mathcal{A}$ based on his input $(i,j)$, as follows:
    \begin{itemize}
        \item $q = C\cdot f(e_{n-W+i}) \in \mathbb{R}^d$, where $C$ is a value to be determined. 
        \item $k,v = \vec{0}$
    \end{itemize}
    We claim that Bob can now recover $x_{ij}$ by examining the sliding window attention output. We have that:
    \begin{align}
        \text{Attn}_W(q,K_{n+1},V_{n+1})_j = \sum\limits_{\ell=n-W+1}^n x_{\ell-n+W,j}\cdot \xi_\ell
    \end{align}
    where $\xi \in \mathbb{R}^n$ is the softmax vector from the values: $(\underbrace{0,0,...,0}_{n-W\text{ values}},\underbrace{q^T k_{n-W+1},...,q^T k_n}_{W \text{ values}})$. We know that for $\ell \neq n-W+i$ we have that:
    \begin{align}
        \xi_\ell \leq e^{C\varepsilon}
    \end{align}
    and that 
    \begin{align}
        \xi_{n-W+i} \geq e^{C(1-\varepsilon)}
    \end{align}
    with probability at least $1-\frac{1}{n}$ because of the guarantees of Lemma \ref{lem:inner-product-jl-for-all}. So now we have two cases:
    \begin{itemize}
        \item $x_{ij} = 0$. We have that:
        \begin{align}
            \text{Attn}_W(q,K_{n+1},V_{n+1})_j &= \frac{\sum\limits_{\ell=n-W+1,\ell\neq n-W+i}^n x_{\ell-n+W,j}\cdot e^{q^T k_\ell}}{n-W+e^{q^T k_{n-W+i}}+\sum\limits_{\ell=n-W+1,\ell\neq n-W+i}^n e^{q^T k_\ell}}\\
            &\leq\frac{\sum\limits_{\ell=n-W+1,\ell\neq n-W+i}^n e^{q^T k_\ell}}{e^{q^T k_{n-W+i}}+\sum\limits_{\ell=n-W+1,\ell\neq n-W+i}^n e^{q^T k_\ell}}
        \end{align}
        We can maximize\footnote{We can be tighter here in the algebra and include $W$, but this is not necessary for determining the value of $C$.}this quantity as:
        \begin{align}
            \text{Attn}_W(q,K_{n+1},V_{n+1})_j \leq \frac{ne^{C\varepsilon}}{e^{C(1-\varepsilon)}+ne^{C\varepsilon}}:=\delta
        \end{align}
        \item $x_{ij} = 1$. We have that:
        \begin{align}
            \text{Attn}_W(q,K_{n+1},V_{n+1})_j &= \frac{e^{q^T k_{n-W+i}}+\sum\limits_{\ell=n-W+1,\ell\neq n-W+i}^n x_{\ell-n+W,j}\cdot e^{q^T k_\ell}}{n-W+e^{q^T k_{n-W+i}}+\sum\limits_{\ell=n-W+1,\ell\neq n-W+i}^n e^{q^T k_\ell}}\\
            &\geq \frac{e^{q^T k_{n-W+i}}}{n-W+e^{q^T k_{n-W+i}}+\sum\limits_{\ell=n-W+1,\ell\neq n-W+i}^n e^{q^T k_\ell}}\\
            &\geq\frac{e^{q^T k_{n-W+i}}}{(n-W)e^{C\varepsilon}+e^{q^T k_{n-W+i}}+\sum\limits_{\ell=n-W+1,\ell\neq n-W+i}^n e^{q^T k_\ell}}
        \end{align}
        We can minimize this quantity as:
        \begin{align}
            \text{Attn}_W(q,K_{n+1},V_{n+1})_j \geq \frac{e^{C(1-\varepsilon)}}{ne^{C\varepsilon}+e^{C(1-\varepsilon)}}:=\Delta
        \end{align}
    \end{itemize}
    Now, because of the approximation guarantees of $\mathcal{A}$, we require:
    \begin{align}
        \delta < \frac{1-\eta}{1+\eta}\Delta \iff \frac{ne^{C\varepsilon}}{e^{C(1-\varepsilon)}+ne^{C\varepsilon}} <\frac{1-\eta}{1+\eta}\cdot \frac{e^{C(1-\varepsilon)}}{ne^{C\varepsilon}+e^{C(1-\varepsilon)}}
    \end{align}
    We have seen that setting:
    \begin{align}
        C > \frac{2\ln n - \ln\frac{1-\eta}{1+\eta}}{2(1-2\varepsilon)}
    \end{align}
    suffices to hold this inequality, with which our proof concludes.
\end{proof} 

\section{Proof of Theorem \ref{thm:time-complexity}}
\label{appx:time-complexity}
In this section we include the full proof of Theorem \ref{thm:time-complexity}. Recall that non-adaptive algorithms are algorithms that pre-decide their data access patterns before seeing the data. 

\begin{theorem}
  Let $Z_n := \operatorname{Attn}(q_n,K_n,V_n)$ and suppose $\mathcal{A}$ is a non-adaptive streaming algorithm for token generation that calculates the attention function in an online manner and outputs an estimate $\mathcal{O} \in \mathbb{R}^d$ such that with probability at least $9/10$ it holds that:
  \begin{align}
      |\mathcal{O}_j-Z_{nj}| \leq \eta\cdot Z_{nj}
  \end{align}
  for $\eta \in (0,1)$ and $j \in [d]$. Then, $\mathcal{A}$ must take $\Omega(nd)$ time to process the last token $(q_n,k_n,v_n)$.
  \end{theorem}
  \begin{proof}
  We start from the expectation-based reformulation of self-attention (Equation \ref{eq:expectation-based-reformulation}):
  \begin{align}
      \text{Attn}(q_n,K_n,V_n) = \E_{\ell \sim D_n}[V_{\ell}]
  \end{align}
  We can treat $\mathcal{A}$ as an algorithm that makes queries to $D_n$ to estimate this expectation. Each query takes time $\Theta(d)$ just to calculate the corresponding score $q_n^T k_\ell$, so it suffices to show that $\Omega(n)$ queries are required to obtain a good approximation with high constant probability.
  
  We use the technique of \textit{indistinguishability}. Suppose that $\mathcal{A}$ makes $o(n)$ queries. Consider the following family $\mathcal{H}$ of input streams for $\mathcal{A}$. The family is parameterized by an index $i$ for which we define:
  \begin{align}
      q_n &= \frac{2}{d}\cdot \ln (n-1)\cdot \mathbf{1}^d,\\
      k_i &= \mathbf{1}^d,\,\text{ and }k_j = \mathbf{0}^d \text{ for }j\neq i,\\
      v_\ell &= \mathbf{1}^d,\,\text{ for }\ell \in [n-\sqrt{n}],\\
      v_\ell &= \sqrt{n}\cdot\mathbf{1}^d\,\text{ for } \ell \in \{\sqrt{n}+1,...,n\}
  \end{align}
  where $\mathbf{1}^d$ is the vector of all $1$s. Note that with this construction the softmax distribution with scores $q_n k_1,...,q_n^T k_n$ has exactly $1-\frac{1}{n}$ of its weight to index $i$ in its support:
  \begin{align}
      e^{q_n^T k_i} = \left(1-\frac{1}{n}\right)\cdot (n-1+e^{q_n^T k_i})
  \end{align}
  Now, consider another input stream $\sigma$ defined as:
  \begin{align}
      q_n &= \mathbf{1}^d,\\
      k_\ell &= \mathbf{0}^d,\quad \forall \ell\in [n]\\
      v_\ell &= \mathbf{1}^d,\,\text{ for }\ell \in [n-\sqrt{n}],\\
      v_\ell &= \sqrt{n}\cdot\mathbf{1}^d\,\text{ for } \ell \in \{\sqrt{n}+1,...,n\}
  \end{align}
  Suppose we pick some $i \in [n]$ uniformly at random. Then, since $\mathcal{A}$ makes $o(n)$ non-adaptive queries to $D_n$, it must output the same value for both inputs $\mathcal{H}(i)$ and $\sigma$ with probability at least $1-1/\text{poly}(n)$, because it queries the index $i$ on which the two input instances differ with probability at most $o(n)/{n} =1/\text{poly}(n)$.
  
  For the input $\sigma$, we know that:
  \begin{align}
      \E_{\ell \sim D_n}[V_{\ell}] = \frac{1}{n}\sum\limits_{i=1}^n v_{i} = \left(2-\frac{1}{\sqrt{n}}\right)\cdot \mathbf{1}^d
  \end{align}
  Meanwhile, for $\mathcal{H}(i)$ we have that $\E_{\ell \sim D_n}[V_{\ell}]$ equals:
  \begin{align}
      &\left(1-\frac{1}{n}\right)\cdot v_i + \frac{1}{n(n-1)} \sum\limits_{j\neq i}v_j\\
      &=\begin{cases}
          \left(1-\frac{1}{n}\right)\cdot \mathbf{1}^d + \frac{2n-\sqrt{n}-1}{n(n-1)},&\text{if } i \in [n-\sqrt{n}]\vspace{2mm}\\
          \left(1-\frac{1}{n}\right)\cdot \sqrt{n}\cdot \mathbf{1}^d + \frac{2(n-\sqrt{n})}{n(n-1)},&\text{otherwise}
      \end{cases}\\
      &=\begin{cases}
          \left(1-\frac{1}{n}\right)\cdot \mathbf{1}^d + \frac{1}{\text{poly}(n)},&\text{if } i \leq n-\sqrt{n}\vspace{2mm}\\
          \left(1-\frac{1}{n}\right)\cdot \sqrt{n}\cdot \mathbf{1}^d + \frac{1}{\text{poly}(n)},&\text{otherwise}
      \end{cases}
  \end{align}
  
  In any case, the two expectations differ by a factor of at least $2 > 1+\eta$ in every dimension. This a contradiction to the approximation guarantees of $\mathcal{A}$, which means that $\Omega(n)$ queries to $D_n$ are required non-adaptively. Thus, the time it takes to process the last token is $\Omega(nd)$. 
  \end{proof}

\end{document}